%% file: main.tex
\documentclass[submission,copyright,creativecommons]{eptcs}
\providecommand{\event}{FMAS 2024} 

\usepackage{iftex}
\usepackage{graphicx}

\ifpdf
  \usepackage{underscore}         
  \usepackage[T1]{fontenc}        
\else
  \usepackage{breakurl}           
\fi
\usepackage{amsmath}
\usepackage[algo2e,linesnumbered,ruled]{algorithm2e}
\usepackage{float}
\usepackage{tikz}
\usepackage{fancyvrb}
\usepackage{xspace}
\usepackage{caption}
\usepackage{tabularx}
\usepackage{amsthm}
\usetikzlibrary{calc, shapes, arrows.meta, positioning, arrows}
\usepackage{soul}
\newcommand{\uline}[1]{\rule[0pt]{#1}{0.4pt}}

\newcommand{\mait}[1]{$\mathit{#1}$~}
\newcommand{\percentage}{\mait{/battery\_percentage}}
\newcommand{\status}{\mait{/battery\_status}}
\newcommand{\statChange}{\mait{/status\_change}}
\newcommand{\inputAcc}{\mait{/input\_accepted}}
\newcommand{\statAcc}{\mait{/status\_accepted}}
\newcommand{\led}{\mait{/SetLED}}
\newcommand{\ledmon}{\mait{/SetLED\_mon}}
\newcommand{\ledtopic}{\mait{/LED\_Panel}}

\newcommand{\maitp}[1]{$\mathit{#1}$}
\newcommand{\percentagep}{\maitp{/battery\_percentage}\xspace}
\newcommand{\statusp}{\maitp{/battery\_status}\xspace}
\newcommand{\statChangep}{\maitp{/status\_change}\xspace}
\newcommand{\inputAccp}{\maitp{/input\_accepted}\xspace}
\newcommand{\statAccp}{\maitp{/status\_accepted}\xspace}
\newcommand{\ledp}{\maitp{/SetLED}\xspace}
\newcommand{\ledmonp}{\maitp{/SetLED\_mon}\xspace}
\newcommand{\verdictp}{\maitp{/verdict}\xspace}

\input{comments}

\title{ROSMonitoring 2.0: Extending ROS Runtime Verification to Services and Ordered Topics}
\author{Maryam Ghaffari Saadat
\institute{University of Manchester \\ Manchester, United Kingdom}
\email{maryam.ghaffarisaadat@manchester.ac.uk}
\and
Angelo Ferrando
\institute{University of Modena and Reggio Emilia \\ Modena, Italy}
\email{angelo.ferrando@unimore.it}
\and
Louise A. Dennis
\institute{University of Manchester \\ Manchester, United Kingdom}
\email{louise.dennis@manchester.ac.uk}
\and
Michael Fisher
\institute{University of Manchester \\ Manchester, United Kingdom}
\email{michael.fisher@manchester.ac.uk}
}
\def\titlerunning{ROSMonitoring 2.0: Extending ROS RV to Services and Ordered Topics}
\def\authorrunning{M. G. Saadat, A. Ferrando, L. A. Dennis, \& M. Fisher}

\begin{document}
\maketitle

\begin{abstract}
Formal verification of robotic applications presents challenges due to their hybrid nature and distributed architecture. This paper introduces ROSMonitoring 2.0, an extension of ROSMonitoring designed to facilitate the monitoring of both topics and services while considering the order in which messages are published and received. The framework has been enhanced to support these novel features for ROS1 -- and partially ROS2 environments -- offering improved real-time support, security, scalability, and interoperability. We discuss the modifications made to accommodate these advancements and present results obtained from a case study involving the runtime monitoring of specific components of a fire-fighting Uncrewed Aerial Vehicle (UAV).
\marginFMASMGS{Made the use of ROS1 consistent in main text (ROS was changed to ROS1 where it wasn't referring to both versions).}
\end{abstract}

\newtheorem{assumption}{Assumption}
\newtheorem{theorem}{Theorem}
\newtheorem{lemma}{Lemma}

\section{Introduction}


The formal verification of robotic applications is a challenging task. Due to their heterogeneous and component-based nature\marginLD{I'm not sure inherent distribution will mean much to the average FM person}\marginMGS{Not sure how to address this.}, establishing the correctness of robotic systems can be particularly difficult. Various approaches exist to tackle this problem, ranging from testing methods~\cite{DBLP:conf/irc/Breitenhuber20,DBLP:journals/jrie/KanterV20,DBLP:journals/sqj/BritoSS22} to static~\cite{DBLP:conf/icfem/FoughaliBDIM16,DBLP:conf/rv/DesaiDS17} or dynamic~\cite{DBLP:conf/rv/HuangEZMLSR14,DBLP:conf/taros/FerrandoC0AFM20} formal verification.
In this work, we focus on the latter approach to verification, specifically the extension of ROSMonitoring~\cite{DBLP:conf/taros/FerrandoC0AFM20}, a Runtime Verification (RV) framework developed for monitoring robotic systems deployed in the Robot Operating System (ROS)~\cite{ros}. ROS is widely used, providing a \emph{de facto} standard for robotic components. ROS encourages component-based development of robotic systems where individual components run in parallel, may be distributed across several processors, and communicate via messages. We tackle ROSMonitoring because it is a novel, formalism-agnostic, and widely used framework for the runtime monitoring of ROS applications. ROSMonitoring allows the specification of formal properties externally to ROS, without imposing any constraints on the formalism to be used. The properties that can typically be monitored in ROSMonitoring concern messages exchanged between different ROS components, called nodes. Such message communication is achieved through a publish-subscribe mechanism, where some nodes (referred to as \emph{publishers}) publish messages on a topic and other nodes (referred to as \emph{subscribers}) subscribe to these topics to listen for the published messages. Through ROSMonitoring, it is possible to specify the communication flow on such topics. For instance, one can determine which messages are allowed in the current state of the system, the correct order amongst them, and other relevant criteria.

Unfortunately, not all aspects of verifying ROS applications are based solely on message communication. In fact, when developing ROS systems, other communication mechanisms can also be utilised, such as \emph{services}. Unlike the publish-subscribe mechanism used with topics, services provide a way for nodes in ROS to directly offer functionalities to each other. While topics are typically used to transmit data from sensors, services serve as an interface that enables nodes to offer specific functionalities to others within the ROS system. Unlike topics, services are commonly synchronous, meaning that when a node calls a service, it waits for a response from the receiving node. This is in contrast to topics, where subscription is non-blocking, and the subscriber node is simply notified whenever a new message is published on the topic, without any waiting involved.
Services are not supported in ROSMonitoring, restricting the framework's functionality to solely monitoring messages.

Another current limitation of ROSMonitoring pertains to the handling of message order. The framework orders messages based on the chronological order in which they are received by subscriber nodes. However, this approach only considers the viewpoint of subscribers, which may not always be suitable. In some scenarios, it may be necessary to consider the order of messages based on when they were sent. For instance, if one message was sent before another, the former should be analysed before the latter by the monitor, appearing earlier in the resulting trace of events. Unfortunately, ROSMonitoring does not currently provide a representation of the order in which messages are published and received. Generally, messages on a single topic are received by subscribers in the order they were published. However, if a property needs to monitor several topics, then it is unusual for the messages from more than one topic to be received in the order they were published. Reordering messages according to publication time is necessary if checking conditional actions that respond to specific event patterns. \marginFMASMGS{added last sentence} \marginLD{I would highlight at this point that, in general, messages on a single topic are received in the order they were sent, but if a property needs to monitor several topics then it is unusual for the messages from more than one topic to be received in the order they were sent.}\marginMGS{Added, thanks!}

In this paper, we introduce ROSMonitoring 2.0, an extension of ROSMonitoring designed to facilitate the monitoring of both topics and services while also considering the order in which messages are published and received. The framework has been enhanced to support these novel features for ROS. Some of these features, \textit{i.e.}, service monitoring, have also been ported to ROS2\footnote{ROS2 is the upgraded version of ROS1, providing improved real-time support, security, scalability, and enhanced interoperability with multiple communication middleware options.} environment as well. We discuss the modifications made to accommodate these advancements and present results obtained from a case study involving the runtime monitoring of specific components of a fire-fighting Uncrewed Aerial Vehicle (UAV).

\section{Preliminaries}


In this section, we briefly introduce Runtime Verification and the ROSMonitoring framework. We emphasise the primary distinction between RV and static verification techniques. Additionally, we provide an overview of the ROSMonitoring framework and briefly outline its main features.

\subsection{Runtime Verification}

Runtime Verification (RV) is a lightweight formal verification technique that checks the behaviour of a system while it is running~\cite{DBLP:journals/jlp/LeuckerS09}. Unlike model checking, RV does not suffer from the state space explosion problem typical in static verification methods and is therefore much more scalable~\cite{DBLP:books/daglib/0007403}. RV is particularly suitable for robotic applications due to resource limitations and system complexity that make full verification at design-time challenging. While static verification techniques focus on abstracting system components, RV checks system behaviour directly. RV addresses the word inclusion problem~\cite{DBLP:journals/scp/AnconaFFM21}, determining if a given event trace belongs to the set of traces denoted by a formal property (referred to as the property's language). This verification process is polynomial in time relative to the trace length. In contrast, model checking exhaustively verifies if a system satisfies or violates a property by analysing all possible system executions, tackling the language inclusion problem, and is typically PSPACE-complete for non-deterministic finite automata~\cite{DBLP:journals/jacm/SistlaC85}. RV commonly employs runtime monitors, automatically synthesised from formal properties, often expressed using Linear-time Temporal Logic (LTL)~\cite{DBLP:conf/focs/Pnueli77}. These monitors gather information from system execution traces and conclude whether the system satisfies or violates the property. A monitor returns $\top$ if the trace satisfies the property, $\bot$ if it violates it, and $?$ if there is insufficient information. Depending on the property's formalism, $?$ may further split into $?_\top$ or $?_\bot$ indicating partial satisfaction or partial violation, respectively.

\subsection{ROSMonitoring}

ROSMonitoring~\cite{DBLP:conf/taros/FerrandoC0AFM20} is a framework for performing RV on ROS applications. ROSMonitoring allows the user to add monitors to ROS applications, which intercept the messages exchanged between components, called ``ROS nodes''\footnote{ROS is node-based, each robot can be composed of multiple nodes.}, and check whether the relevant messages conform to a given formal property. In the following we describe these three different aspects in more detail.

\input{rosmonoverview}

\subsubsection{Instrumentation}

ROSMonitoring starts with a YAML configuration file to guide the instrumentation process required to generate the monitors. Within this file, the user can specify the communication channels, called ``ROS topics'', to be intercepted by each monitor. In particular, the user indicates the name of the topic, the ROS message type expected in that topic, and the type of action that the monitor should perform. After preferences have been configured in \emph{config.yaml}, the last step is to run the generator script to automatically generate the monitors and instrument the required ROS launch files.

\subsubsection{Oracle}
\label{sec:oracle}
ROSMonitoring decouples the message interception (monitor) and the formal verification aspects (oracle) and so is highly customizable. Different formalisms can be used to represent the properties to be verified, including Past MTL, Past STL, and Past LTL (MTL~\cite{DBLP:journals/rts/Koymans90}, STL~\cite{DBLP:conf/formats/MalerN04}, and LTL~\cite{DBLP:conf/focs/Pnueli77} with past-time operators, respectively). Using the formalism of choice, an external entity can be created to handle the trace of events reported by the monitors in ROS (generated through instrumentation). ROSMonitoring requires very few constraints for adding a new oracle. It uses JSON\footnote{\url{https://www.json.org/}} (JavaScript Object Notation) as a data-interchange format for serialising the messages that are observed by the ROS monitor. JSON is commonly used for transmitting data between a server and a web application. In JSON, data is represented as key-value pairs enclosed in curly braces, making it a popular choice for APIs and data storage. An oracle will parse the JSON messages, check whether they satisfy or violate the formal property, and report back to the ROS monitor.

\subsubsection{ROS monitor}

The instrumentation process generates monitors to intercept the messages of interest. Each monitor is automatically generated as a ROS node in Python, which is a native language supported in ROS. ROSMonitoring provides two types of monitors: 1) offline monitors which simply log the intercepted events in a specified file to be parsed by the Oracle later to determine whether they satisfy a given set of properties, and 2) online monitors which query the Oracle in real time about whether the intercepted messages satisfy the given properties. While offline monitors only log the observed messages, online monitors could either log messages along with the Oracle verdict updated after each message or filter messages that the Oracle deems have violated the given properties. To clarify the difference, in the case of logging without filtering, if the online monitor finds a violation of the property under analysis, it publishes a warning message containing as much information as possible about the violated property. This warning message can be used by the system to handle the violation and to react appropriately. However, the monitor does not stop the message from propagating further in the system. In contrast, if filtering is enabled, since monitors can be placed between the communication of different nodes, ROSMonitoring monitors enforces the property under analysis by not propagating messages that represent a property violation. This is achieved by directing communication on the monitored topics to pass through the monitors. \marginFMASMGS{added explanation about offline and online monitors and their actions.}

\section{Motivating example}

In this section, we explain the rationale behind extending the ROSMonitoring framework. We use, as an example, a Battery Supervisor system \footnote{Full code for this example is available in the \href{https://github.com/LilithMary/ROSMonitoring2.0-Case-Study.git}{Git Repository for ROSMonitoring 2.0 Case Study}.} designed for a UAV (Uncrewed Aerial Vehicle) with three essential components depicted in Figure \ref{fig:example}: the Battery, the Battery Supervisor, and the LED Panel. The Battery periodically reports the remaining battery percentage. The Battery Supervisor is responsible for checking the battery level and reporting its status. It subscribes to the battery percentage updates and analyses them. If the battery percentage is above 40\%, it signals a `healthy' status. If it is between 30\% and 40\%, it flags a `warning' status. And if it falls below 30\%, it indicates a `critical' status. The LED Panel reflects the battery status through coloured LED lights. The Battery Supervisor is connected to the LED Panel and whenever it detects a change in the battery status, it sends a signal to the LED Panel to adjust the lights accordingly. For instance, if the battery is in a critical state, the red light might flash to indicate urgency\footnote{This case study was inspired by the example for RS services in the \href{https://roboticsbackend.com/what-is-a-ros-service/}{Robotics Back-End Tutorial} as well as a solution to Challenge 3 of the MBZIRC Challenge competition 2020~\cite{mbzirc2020}.}.\marginLD{Even though we are a long way from the original version of this system, it might be useful to mention this is inspired by a solution to Challenge 3 of the MBZIRC Challenge competition 2020 - it implies there is a real use case driving the developments and example.}\marginMGS{Added footnote.} \marginFMASMGS{Added link to Git repository for case study (and added description to the repository). Added link to example which inspired our case study.}

\begin{figure}[!ht]
    \centering
    \resizebox{.8\textwidth}{!}{
    \input{example}
    }
    \captionof{figure}{Motivating example with three components: Battery, Battery Supervisor, and LED Panel. Battery publishes on topic \percentagep; Battery Supervisor subscribes to \percentagep, publishes on topic \statusp, and invokes service \ledp; LED Panel publishes on topic \ledtopic and responds to \led service requests.}
    \label{fig:example}
\end{figure}

In this example, we are interested in ensuring that the messages exchanged between different components correspond correctly. For instance, that every status update provided by the Battery Supervisor accurately reflects the current battery percentage received from the Battery. However, while messages published on a single topic generally arrive at the subscribers according to their publication order, messages on different topics can arrive out of order. As a result, a battery status message could be observed before its corresponding battery percentage. Therefore, we need additional mechanisms to account for this before sending the messages to the Oracle for verification. This motivates our extension to reorder the messages according to their publication time. \marginLD{Stress here that the issue is that these messages are published on different topics.}\marginMGS{Changed the wording at the beginning of the paragraph to address this.}

Beyond message correspondence, we are also interested in verifying the interaction between the Battery Supervisor and the LED Panel service. We would like to confirm that every time the LED lights are adjusted based on the battery status, it is triggered by a legitimate status update. Conversely, we would like to check that every change in battery status is promptly followed by a request to adjust the LED lights, maintaining synchronisation between the visual feedback and the actual battery condition. However, services are not supported by the ROSMonitoring framework. This motivates our extension to support services. Since some of the properties we are interested in monitoring include both topics and services, we have also developed support for reordering service requests and responses according to publication time as well.

\section{ROSMonitoring 2.0}

ROSMonitoring 2.0 is fully available\footnote{\url{https://github.com/autonomy-and-verification-uol/ROSMonitoring/tree/master}} for ROS1, while only partially available\footnote{\url{https://github.com/autonomy-and-verification-uol/ROSMonitoring/tree/ros2}} for ROS2 (service monitoring has been added but not message reordering). In this section, we present two novel aspects of ROSMonitoring 2.0. Firstly, in Section \ref{subsection:service-extension}, we detail the mechanism enabling monitoring of ROS services. Secondly, in Section \ref{subsection:reordering-extension}, we introduce an algorithm for reordering messages based on their publication time, demonstrating its correctness under the assumption that messages on each topic arrive sequentially. Notably, such reordering mechanism is extended to support services in addition to topics.

\subsection{Service extension}
\label{subsection:service-extension}
While ROS topics excel in broadcasting data streams or events asynchronously to multiple nodes, ROS services are designed for synchronous, point-to-point communication to request specific actions or services from other nodes in the system. Consequently, when monitoring services, our monitor node must directly intervene in the communication between the server and client. The monitor node then assumes the role of a server for the client, and conversely acts as a client for the server. The sequence diagram in Figure \ref{fig:rosmon-service-verification} illustrates the service verification process in ROSMonitoring 2.0 where message filtering is enabled (which subsumes the non-filtering scenario). The scenario begins with the Client sending a service request $\mathit{callService(req, res)}$ to the Monitor. Subsequently, the Monitor forwards the request to the Oracle for verification via a callback mechanism specific to the service. If the Oracle identifies the request as inconsistent with the defined property, it responds with a negative verdict (\textit{i.e.} either $?_\bot$ or $\bot$)). In response, the Monitor publishes an error message and notifies the client of the discrepancy, bypassing the service invocation. Conversely, if the Oracle confirms the consistency of the request with the property, it returns a positive verdict (\textit{i.e.}, either $?_\top$ or $\top$). The Monitor proceeds to invoke the service and awaits a response from the server. Upon receiving the response, the Monitor relays it back to the Oracle for evaluation. Should the Oracle determine the response to be erroneous (\textit{i.e.}, the returned verdict is either $?_\bot$ or $\bot$), the Monitor again publishes an error message and notifies the client accordingly. Otherwise, it delivers the response to the client as expected.

In contrast to the standard ROSMonitoring behaviour, handling services necessitates additional verification steps. The Monitor must check both the service request and its corresponding response with the Oracle. Verifying the request is crucial to prevent invoking the service in case of a violation. Moreover, the Monitor must act as an intermediary between the client and server. This mechanism mirrors ROSMonitoring's behaviour when topic filtering is enabled, albeit with an extension in the case of services to invoke the actual service upon successful request verification.


\input{seq-diagram}

\marginAF{Fixed sequence diagram.}

\subsection{Ordered topics extension} 
\label{subsection:reordering-extension}
To recover the publication order of messages in real time, ROSMonitoring 2.0 adds timestamps to each message. These timestamps are then utilised in Algorithm~\ref{alg:ordered-msgs} to propagate messages to the Oracle in their original publication order. Moreover, this approach is based on the following assumption.
\marginAF{Algorithm 1 is now mentioned at the beginning of Section 4.2.}

\begin{assumption}\label{assumption:message-order-single-topic}
Messages on each single topic arrive at subscribers in the order of publication.
\end{assumption}

\begin{center}
{\small
\begin{algorithm2e}[!ht]
    \SetKwInOut{Input}{Input}
    \SetKwInOut{Output}{Output}
    \SetKwFunction{addToBuffer}{addToBuffer}
    \SetKwFunction{sendEarliestMessageToOracle}{sendEarliestMessageToOracle}
    \SetKwProg{Fn}{Function}{:}{}
    \Input{
    \newline
    \begin{tabularx}{\textwidth}{l p{6cm}}
    $\mathit{msg}$: & a ROS message received by the monitor on a topic \textit{t} \\
    $\mathit{ws}$: & a global websocket \\
    $\mathit{buffer}$: & a global dictionary mapping each topic to a list of timestamps of unprocessed
    messages published on that topic \\
    $\mathit{messages}$: & a global dictionary mapping publication timestamp to corresponding message\\
    \end{tabularx}
    }
    \Output{Verdict published by Oracle based on messages in the order of publication time}
    \vspace{.2cm}
    \Fn{\sendEarliestMessageToOracle{}}
    {
    using $\mathit{buffer}$ and $\mathit{messages}$ \vspace{.3cm}\\
    \begin{tabularx}{\textwidth}{l l p{6cm}}
    $\mathit{min\_time\_stamp}$ & $=$ & minimum timestamp in $\mathit{messages}$ dictionary\\
    $\mathit{message}$ & $=$ & $\mathit{messages[min\_time\_stamp]}$\\
    \end{tabularx}
    send \textit{message} to Oracle\\
    $\mathit{verdict =}$ Oracle's response\\
    remove $\mathit{min\_time\_stamp}$ from $\mathit{buffer}$\\
    remove $\mathit{message}$ from $\mathit{messages}$\\
    Publish $\mathit{verdict}$ \\
    
    }
     \vspace{.2cm}
    \Fn{\addToBuffer{msg, t}}
    {
     using $\mathit{ws}$, $\mathit{buffer}$, and $\mathit{messages}$\vspace{.3cm}\\

    $\mathit{time\_stamp\_of\_msg = getTime(msg)}$\\
    add $\mathit{time\_stamp\_of\_msg}$ to $\mathit{buffer[t]}$\\
    $\mathit{messages[time\_stamp\_of\_msg] = msg}$

    lock websocket $\mathit{ws}$\\
    \While{no topic has an empty buffer}{
    
    \sendEarliestMessageToOracle()\\
    
    }
    unlock websocket \textit{ws}

    }
\caption{Algorithm for propagating messages to the Oracle in the order they were published}\label{alg:ordered-msgs}
\end{algorithm2e}
}
\end{center}

In order to use the reordering feature in the ROS monitor, in the callback function for every topic, instead of propagating the message ($msg$) directly to the Oracle, Algorithm~\ref{alg:ordered-msgs} calls $\mathit{addToBuffer(msg, t)}$. Such a procedure accumulates messages from each topic into their respective buffers. Message release is withheld until all buffers contain at least one message, at which point the message with the earliest publication timestamp is released and sent to the Oracle. To prevent more than one thread to change the buffers simultaneously, we use locks to block write-access for a single thread. \\
\marginFMASMGS{brought Algorithm 1 ahead of its explanation and Lemma 1. This also fixed the spacing issue on the previous version.}

\marginLD{I'd write some more here to motivate the idea behind the algorithm in terms of buffering the messages on each topic being monitored and only releasing messages when there is at least one message in each buffer.  Expecting a reader to intuit this from the pseudocode is a lot}\marginMGS{Added a couple of sentences.}

\begin{lemma}\label{lemma:buffer-order}
    In Algorithm~\ref{alg:ordered-msgs}, messages in each buffer maintain the order of publication timestamps.
\end{lemma}

\begin{proof}
    Suppose there is a topic $t$ such that its corresponding list in dictionary $\mathit{buffers}$ is out of order. Without loss of generality, assume there are two messages $m_1$ and $m_2$ on topic $t$ with $m_1$ published before $m_2$ but stored in $\mathit{buffers[t]}$ in reverse order, \textit{i.e.} $[.. , m_2, m_1, ..]$. Due to Assumption~\ref{assumption:message-order-single-topic}, since both $m_1$ and $m_2$ are on the same topic, they are received by the ROS monitor in the correct order. Therefore, $\mathit{addToBuffer(m_1, t)}$ is called before $\mathit{addToBuffer(m_2, t)}$. Consequently, $m_1$ is appended to the list $\mathit{buffers[t]}$ before $m_2$. This contradicts our assumption that $m_2$ is stored before $m_1$ in $\mathit{buffers[t]}$. Thus, it follows, by contradiction, that Lemma~\ref{lemma:buffer-order} holds. 
\end{proof}

\begin{theorem}
    Algorithm~\ref{alg:ordered-msgs} propagates messages to the Oracle in the order of their publication.
\end{theorem}

\begin{proof}
In order to prove that Algorithm~\ref{alg:ordered-msgs} is correct, we assume the opposite, namely that two messages, $m_1$ and $m_2$, were propagated to the Oracle in reverse order of their publication timestamps. Without loss of generality suppose $m_1$ was published earlier than $m_2$ but was propagated to the Oracle after $m_2$. 

If the messages are on the same topic then, due to Assumption~\ref{assumption:message-order-single-topic}, we reach a contradiction which means our assumption is incorrect and the proof is complete. Otherwise, the messages are on distinct topics. Hence, by construction, they are stored in separate lists in \textit{buffers}. Furthermore, the algorithm only sends messages to the Oracle if the buffers for all topics are non-empty. Therefore, both $m_1$ and $m_2$ must be present in their corresponding buffers at the time $m_2$ is propagated to the Oracle. But the algorithm, by construction, always chooses the message with the smallest timestamp to propagate to the Oracle next. This contradicts our assumption that $m_2$ is propagated before $m_1$ despite having a larger timestamp. Consequently, we can conclude, by contradiction, that Algorithm~\ref{alg:ordered-msgs} is correct.  
\end{proof}

A successful application of the ordering mechanism necessitates careful consideration to mitigate the risk of deadlocks. Specifically, when a topic $t$ undergoes filtering by the monitor, and another topic or service $x$ relies on it, both $t$ and $x$ should not be concurrently included in the ordering process. Otherwise, the buffering of a message $m_1$ on topic $t$ can lead to a deadlock scenario, as it cannot be released until an $x$ message is buffered. Conversely, an $x$ message cannot be generated until a $t$ message is published, which, in turn, cannot occur until $m_1$ is released from the buffer. Furthermore, it is essential to carefully evaluate dependencies between topics and services when determining which should be ordered based on their publication times.



\section{Experimental evaluation}

In our case study, we illustrate the practical implementation of ROSMonitoring 2.0 through a scenario involving a Battery Supervisor system for a UAV. We developed an online ROS monitor which runs alongside the system and checks its behaviour against a set of properties in real time. The experiments were conducted using \href{https://wiki.ros.org/noetic}{ROS1 Noetic distribution}. \marginFMASMGS{Commented out Humble, added two sentences to say we use online RV and Noetic ROS.}
As shown in Figure \ref{fig:example}, our case study comprises three interconnected nodes: the Battery, responsible for publishing the remaining battery percentage; the Battery Supervisor, which subscribes to the battery percentage topic and publishes status updates based on predefined thresholds; and the LED Panel, which reflects the battery status through LED lights.

The Battery node periodically broadcasts the battery percentage on the \percentage topic. Meanwhile, the Battery Supervisor node, operating at a slower rate than the Battery, subscribes to this topic and publishes status updates on the \status topic. These status updates indicate the battery status as follows: status 1 for a percentage higher than 40\%, status 2 for a percentage between 30\% and 40\%, and status 3 for a percentage between 0\% and 30\%. Since the Battery Supervisor has a slower publication rate, it may not report the status for every percentage published by the Battery. This is intentional to ensure that while the battery status is reported regularly, energy usage and communications are optimised. 

As shown in Figure \ref{fig:casestudy}, to facilitate monitoring system behaviour, supplementary topics and a service are added to the original example in Figure \ref{fig:example}. For instance, to ensure synchronisation between the Battery and the Battery Supervisor, an additional topic \inputAcc is introduced. This topic tracks which battery percentage messages have been processed by the Battery Supervisor.
Furthermore, the Battery Supervisor publishes a message on topic \statChange if the battery status changes. The Battery Supervisor node subscribes to the \statChange topic itself to separate the processing of \percentage messages from the invocation of service call to \led. As explained further below, this is a workaround to prevent deadlocks when using the ordering mechanism. Upon detecting a change in status by comparing the current status with the previous one, the Battery Supervisor publishes a \statChange message. Once the Battery Supervisor receives a \statChange message, it calls the \ledmon service to update the LED lights on the LED Panel accordingly. After receiving a \ledmon service request, the ROS monitor checks that the request is valid and it calls the \led service. \marginLD{Latex macro weirdness - make sure there isn't a space before the comma} Upon receiving a service call, the LED Panel publishes a message on the \statAcc topic to record which status update was acknowledged, followed by a message on the \mait{/LED\_panel} topic reporting the current state of the LED lights (green, yellow, and red). The LED Panel also sends a response to the ROS monitor which is relayed back to the Battery Supervisor.  

Subscribing to the \statChange topic may appear peculiar for the Battery Supervisor, which publishes it. While it might seem more straightforward to invoke the \led service where the status calculation occurs based on received \percentage messages, such an approach risks deadlock. The reason is that the buffer is unable to release a service request message until accepting the next percentage message. But no further percentage messages can be accepted until a service response is received and that can only happen if the service request is released. To resolve this, we separated the publication of topics from the function which initiates service requests so that the service does not block the receipt of messages needed for producing a response. 

\begin{figure}[!ht]
    \centering
    \resizebox{.8\textwidth}{!}{
    \input{casestudy}
    }
    \captionof{figure}{Case study with ROS Monitor and additional topics \inputAccp, \statAccp, \statChangep, \verdictp and service \ledmonp.}
    \label{fig:casestudy}
\end{figure}

\medskip


\noindent The properties we selected to verify are as follows with formal definitions in Table \ref{tab:properties}:
\begin{enumerate}
    \item \textbf{Topic only:} Correspondence of \status with \percentage and \inputAcc:
    \begin{enumerate}
    \item Every \status message corresponds to a \inputAcc message and correctly reports the status based on its corresponding \\ 
    \percentage message.\label{property:topic-1}
    \item Every \inputAcc message is followed by a \status message within 100 time steps. \label{property:topic-2}
    \end{enumerate}
    \item \textbf{Topic and Service:} Correspondence of \led service request with\\ \status:
    \begin{enumerate}
    \item Every \led service request corresponds to a \status message and a change in battery status. \label{property:topic-service-1}
    \item Every change of status reported via \status messages is followed by a \led service request within 100 time steps. \label{property:topic-service-2}
    \end{enumerate}
    \item \textbf{Service only:} Correspondence of \led service request and response:
    \begin{enumerate}
    \item Every \led service response corresponds to a \led service request. \label{property:service-1}
    \item Every \led service request is followed by a \led service request within 100 time steps. \label{property:service-2}
    \end{enumerate}    
\end{enumerate}

\begin{table}[t]
    \centering
    \resizebox{0.8\textwidth}{!}{
    \begin{tabular}{|l |l|}\hline
    Property ID & Property formal specification\\\hline
        \ref{property:topic-1}  & forall[i]. (forall[s]. \{topic: ``\statusp'', id: *i, status: *s\}\\
        & $\rightarrow$  once(\{topic: ``\inputAccp'', id: *i\}) and \\
          & ~~~~once(\{topic: ``\percentagep'', id: *i, percentage: *s\}))\\\hline
        \ref{property:topic-2} & forall[i]. not (\{topic: ``\statusp'', id: *i\})\\ 
        & $\rightarrow$  once[1:]\{topic: ``\statusp'', id: *i\} or not (once[100:](\{topic: ``\inputAccp'', id: *i\}))\\\hline
        \ref{property:topic-service-1} & forall[i]. (forall[s]. \{service: ``\ledp'', req\_id: *i, req\_status: *s\}\\ 
        & $\rightarrow$ once(\{topic: ``\statusp'', id: *i, status: *s, status\_change: True\})\\\hline
        \ref{property:topic-service-2} & forall[i]. not (\{service: ``\ledp'', req\_id: *i, req\_status: *s\})\\ 
        & $\rightarrow$ once[1:] \{service: ``\ledp'', req\_id: *i, req\_status: *s\} \\
        & ~~~ or not (once[100:](\{topic: ``\statusp'', id: *i, status: *s, status\_change: True\})\\\hline
        \ref{property:service-1} & forall[i]. \{service: ``\ledp'', response: True, res\_id: *i\}\\ 
        & $\rightarrow$ once(\{service: ``\ledp'', request: True, req\_id: *i\})\\\hline
        \ref{property:service-2} & forall[i]. not (\{service: ``\ledp'', response: True, res\_id: *i\}) \\
        & $\rightarrow$ once[1:] \{service: ``\ledp'', response: True, res\_id: *i\} \\
        & ~~~ or not (once[100:](\{service: ``\ledp'', request: True, req\_id: *i\}))\\\hline
    \end{tabular}
    }
    \captionof{table}{Properties in Past Metric Temporal Logic (Past MTL) according to the Reelay Expression Format (\url{https://doganulus.github.io/reelay/rye/}).}
    \label{tab:properties}
\end{table}

\noindent Formalisation of these properties requires definition of predicates, summarised in Table \ref{tab:predicates}, based on the JSON\marginLD{I think this is the first mention of JSON}\marginMGS{Added a line of introduction to JSON in \ref{sec:oracle}} messages sent to the Oracle. For topics, generic predicates $\mathit{topic}$ and $\mathit{id}$ are defined. Additionally, for the \percentage topic, predicate $\mathit{percentage}$ is defined, taking values $1$ for percentages between 40\% and 100\%, $2$ for percentages between 30\% and 40\%, $3$ for percentages between 0\% and 30\%, and `$\mathit{INVALID}$' for other values. Consistently, the $\mathit{percentage}$ and $\mathit{status}$ predicates share values to enable referencing, as seen in Property~\ref{property:topic-1}. For the \status topic, predicate $\mathit{status}$ holds the String version of message status, with `$\mathit{INVALID}$' assigned if the status is not $1$, $2$, or $3$. Additionally, predicate $\mathit{status\_change}$ is defined to be `True' if the corresponding field in the message is `true'. Note that the monitor is not subscribed to the $\mathit{/status\_change}$ topic which is used to trigger an LED Panel response. The reason for this redundancy is deadlock prevention. If the $\mathit{/status\_change}$ topic was ordered, then its release would be contingent on a service request joining the buffers which cannot happen unless the $\mathit{/status\_change}$ message is released. Our solution to this potential deadlock was to keep the $\mathit{/status\_change}$ topic unordered and add a field $\mathit{status\_change}$ to the $\mathit{battery\_status}$ topic for the Oracle to determine if the LED panel is responding correctly. Such redundancies could be considered as a general technique to prevent deadlocks. For the \led service, predicates $\mathit{request}$ and $\mathit{response}$ indicate message type. For \led service request messages, predicates $\mathit{req\_id}$ and $\mathit{req\_status}$ store additional information used in Properties~\ref{property:topic-service-1} and~\ref{property:topic-service-2} to verify legitimate requests triggered by corresponding \status messages.

\begin{table}[t]
    \centering
    \resizebox{0.8\textwidth}{!}{
    \begin{tabular}{|l|l|l|}\hline
       Message Type  & Predicate & Description \\\hline
       \percentage &  $\mathit{topic}$ & Topic of the message, i.e. \percentage\\\cline{2-3}
               & $\mathit{id}$ & Unique sequentially assigned ID for the message\\\cline{2-3}
                 & $\mathit{percentage}$ & `1' if percentage $>$ 40 and percentage $\leq$ 100\\
                & & `2' if percentage $>$ 30 and percentage $\leq$ 40\\
                & & `3' if percentage $\leq$ 30 and percentage $\geq$ 0\\
                & & `INVALID' if percentage $<$ 0 or percentage $>$ 100\\\hline
        \status  &  $\mathit{topic}$ & Topic of the message, i.e. \status\\\cline{2-3}
               & $\mathit{id}$ & The ID of the corresponding percentage message\\\cline{2-3}
               & $\mathit{status}$ & '0' if status is 0\\
        & & `1' if status is 1\\
        & & `2' if status is 2\\
        & & `3' if status is 3\\
        & & `INVALID' if status is not 0, 1, 2, or 3\\\cline{2-3}
        &  $\mathit{status\_change}$ & `True' if and only if the corresponding field is `true'\\\hline
        \led & $\mathit{request}$ & `True' if and only if the message is a service request for \led\\\cline{2-3}
        & $\mathit{req\_id}$ & ID of the corresponding \status message \\\cline{2-3}
        & $\mathit{req\_status}$ & Status of the corresponding \status message\\\cline{2-3}
        & $\mathit{response}$ & `True' if and only if the message is a service response for \led\\\hline
    \end{tabular}
    }
    \captionof{table}{Predicates construction based upon JSON messages sent to Oracle.}
    \label{tab:predicates}
\end{table}

A significant concern associated with the implementation of runtime verification is its potential adverse effect on overall system performance. To gauge the extent of overhead induced by monitoring services, we modified the client node to measure the time elapsed between dispatching a \led service request and receiving the corresponding response. All experiments shown in Figures~\ref{topic-errorbars}, \ref{fig:service-plot-monitoring}, and~\ref{tab:service-plots-ordering} were conducted over 10 runs, with the results averaged. Each run was terminated once the battery percentage reached zero. Frequencies for the Battery, Battery Supervisor, and LED Panel were set to 25, 10, and 35 Hertz respectively. The mean and standard deviation are reported for each experiment.\marginAF{Added note on experiments being sampled on 10 runs.}\marginFMASMGS{Moved up the setup information for experiments and added frequencies.} As illustrated in Figure \ref{fig:service-plot-monitoring}, the overhead incurred by monitoring \led without ordering appears negligible, but the introduction of ordering substantially delays the process, particularly noticeable during the last status change. To approximate the overhead attributed to the ordering mechanism, we adapted the monitor code to record the time difference between message buffering and transmission to the Oracle (reported in Figure~\ref{tab:service-plots-ordering}). The results depicted in Figure \ref{topic-errorbars} exhibit a similar trend, wherein the release time deviates from the buffering time until the second service request, after which the waiting time stabilises at a minimised level. This is because after the initial service request and response at the beginning of the execution, messages keep accumulating in the buffers since the \led service buffer remains empty until the second service request is triggered by battery percentage reaching 40\%. This threshold is reached when the message with ID 60 is sent. With a service request in the corresponding buffer, the algorithm proceeds to release messages from buffers in the order of publication. Since the next service request is triggered when the battery percentage reaches 30\%, i.e. message ID 70, the service buffer will not be empty for long. This explains the minimal waiting time between buffering and transmission after the 60th message. Once the battery percentage reaches 0, the execution is interrupted. This allows the remaining messages in the buffers to be released in the order of publication without requiring all buffers to be nonempty. Moreover, as shown in Figure \ref{tab:service-plots-ordering}, the waiting time for \led service requests escalates over time, whereas response messages appear to be promptly released from the buffer. This is because at the time that a service request is buffered, a number of messages have accumulated in the other buffers which need to be released before the service request. In contrast, since the service response is buffered shortly after the corresponding service request is released, there are only a few messages buffered in between which need to be released before the service response.  

\begin{table}[!ht]
    \centering
    \resizebox{1\textwidth}{!}{
    \begin{tabular}{l}
        
    \includegraphics{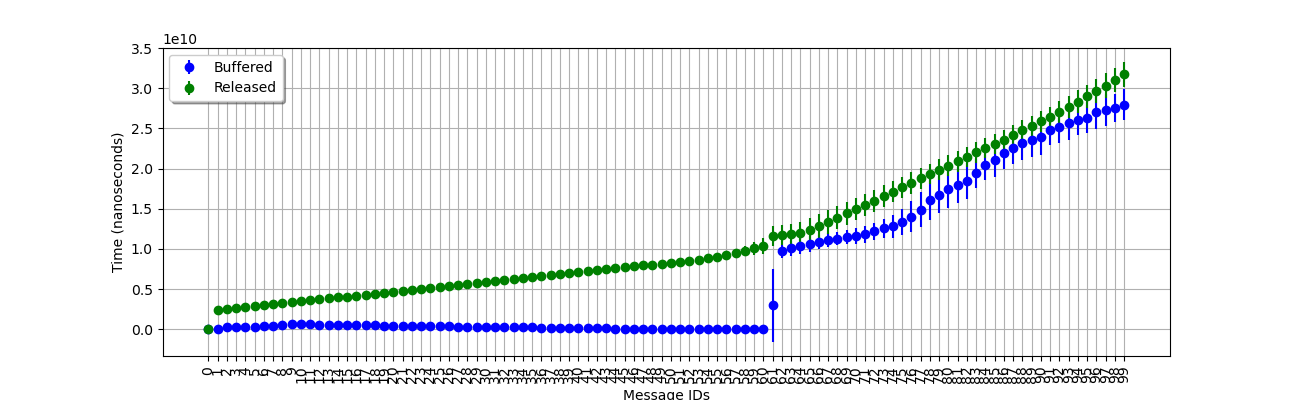}
    
         \\
   
    \includegraphics{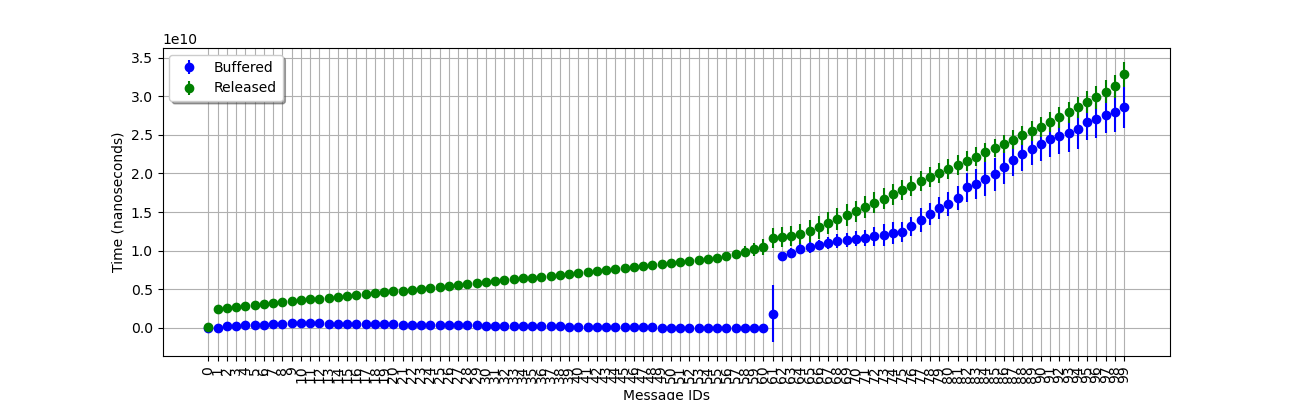}
        
         \\
    \includegraphics{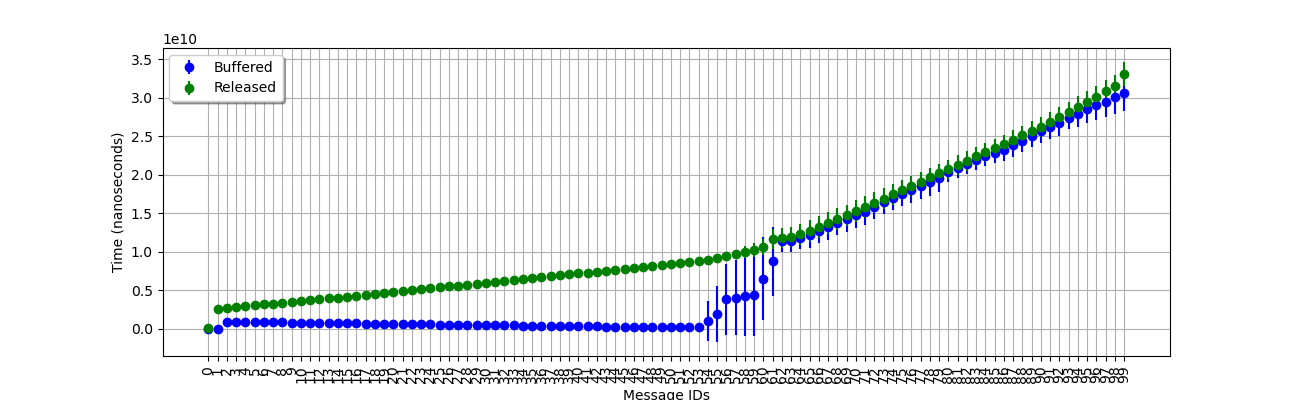}
        
         \\
         
    \end{tabular}
    }
    \captionof{figure}{Mean and standard deviation of time (nanoseconds) taken for monitor with ordering to buffer and release messages on \percentage (top), \inputAcc (middle), ans \status (bottom) topics since publication. Data shown for 10 runs.}
    \label{topic-errorbars}
\end{table}

With regards to verification accuracy, monitoring with ordering consistently yielded accurate verdicts without any incorrect assessments. Conversely, in cases where monitoring excluded ordering, \status messages often reached the monitor prior to their corresponding \percentage or \inputAcc signals. Similarly, service request and response messages consistently preceded the corresponding \statChange signals. These out-of-order message arrivals led to frequent false negative verdicts in each run. Hence, in weighing the trade-off between performance and accuracy, it becomes evident that monitoring with ordering is most suitable for safety-critical systems where time sensitivity is not paramount. On the other hand, monitoring without ordering may offer enhanced performance at the expense of accuracy, making it more suitable for scenarios where real-time constraints are less stringent.

\begin{figure}[!ht]
    \centering
\resizebox{0.5\textwidth}{!}{
    \includegraphics{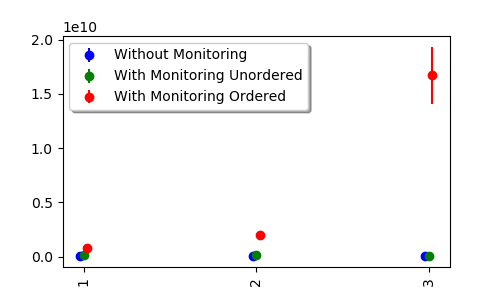}
    }
    \captionof{figure}{Mean and standard deviation of time (nanoseconds) taken for \led service request to receive a response (reported on three service calls). Data shown for 10 runs without monitoring, with monitoring excluding ordering, and with monitoring including ordering.}
    \label{fig:service-plot-monitoring}
\end{figure}




\begin{table}[t]
    \centering
    \resizebox{1\textwidth}{!}{
    \begin{tabular}{ll}
        
    \includegraphics{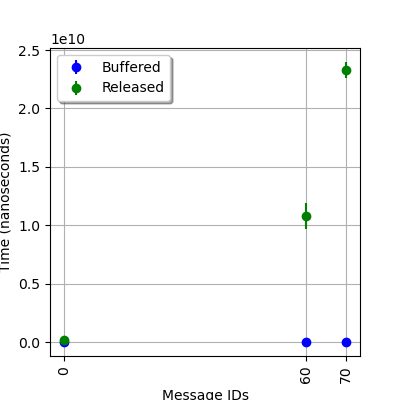}
    
         &  
   
    \includegraphics{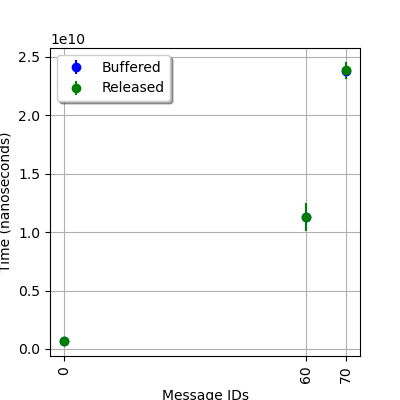}
        
         \\
    \end{tabular}
    }
    \captionof{figure}{Mean and standard deviation of time (nanoseconds) taken for monitor with ordering to buffer and release \led request (left) and response (right) messages since generation. Time for buffering and releasing overlap for response messages in the right plot. Data shown for 10 runs.}
    \label{tab:service-plots-ordering}
\end{table}




\section{Related Work}

In this section, we discuss the most recent approaches to RV of ROS and position them in relation to ROSMonitoring 2.0.

ROSRV~\cite{DBLP:conf/rv/HuangEZMLSR14} shares similarities with our framework in achieving automatic RV of applications in ROS. Both tools utilise monitors not only to passively observe but also to intercept and handle incorrect behaviours in message exchanges among nodes. The main difference lies in how they integrate the monitor into the system. ROSRV replaces the ROS Master node with RVMaster, directing all node communication through it and establishing peer-to-peer communication with the monitor as the intermediary. In contrast, ROSMonitoring adds the monitor through node instrumentation without altering the ROS Master node. Additionally, the new and extended version ROSMonitoring 2.0 presented in this paper further differentiates the two frameworks, as ROSRV does not support the verification of services or the customisation of the order of topics. HAROS~\cite{DBLP:conf/iros/SantosCML16} is a framework dedicated to ensuring the quality of ROS systems. Although HAROS primarily focuses on static analysis, it possesses the capability to generate runtime monitors and conduct property-based testing. Differently from ROSMonitoring 2.0, HAROS does not support ROS2 (not even partially). Furthermore, one notable distinction between ROSMonitoring 2.0 and HAROS is that our specifications do not incorporate ROS-specific details, and the process for generating monitors does not rely on understanding the topology of the ROS graph. DeROS~\cite{DBLP:conf/simpar/AdamLJS14} is a domain-specific language and monitoring system tailored for ROS. Although DeRoS's language incorporates explicit topic notions, it lacks native support for reordering or service handling. Moreover, it is exclusively compatible with ROS. An extension of Ogma supports the runtime monitoring of ROS2 applications~\cite{DBLP:journals/corr/abs-2209-14030}. It outlines a formal approach to generate runtime monitors for autonomous robots from structured natural language requirements, expressed in FRET~\cite{DBLP:conf/refsq/GiannakopoulouP20a}. This extension integrates FRET and Copilot~\cite{copilot} via Ogma to translate requirements into temporal logic formulas and generate monitor specifications. Unlike ROSMonitoring 2.0, which focuses on monitoring and potentially filtering topics and services, this extension is limited to detecting and reporting violations only. Nonetheless, \cite{DBLP:conf/refsq/GiannakopoulouP20a} provides a lightweight verification solution for complex ROS2 applications, ensuring safe operation. MARVer~\cite{DBLP:conf/taros/DegirmenciKOBKOY23} is an integrated runtime verification system designed to ensure the safety and security of industrial robotic systems. It offers a lightweight yet effective approach to monitoring the behaviour of robotic systems in real-time, enabling the detection of security attacks and potential safety hazards. By being based on ROSMonitoring, MARVer can leverage the new features introduced in this work.

The work in~\cite{DBLP:conf/ictac/BegemannKLS23} introduces TeSSLa-ROS-Bridge, a RV system designed for robotic systems built in ROS. Unlike other RV approaches, TeSSLa-ROS-Bridge utilises Stream-based Runtime Verification (SRV), which specifies stream transformations to detect errors and control system behaviour (currently supported in ROSMonitoring as well). The system allows TeSSLa monitors to run alongside ROS-based robotic systems, enabling real-time monitoring. Compared to ROSMonitoring, which focuses on monitoring and filtering topics and services, TeSSLa-ROS-Bridge offers a different approach by leveraging stream-based runtime verification to monitor and control robotic systems. RTAMT is an online monitoring library for Signal Temporal Logic (STL), supporting both discrete and dense-time interpretations. In~\cite{DBLP:conf/atva/Nickovic020}, RTAMT4ROS is introduced, integrating RTAMT with ROS. This integration enables specification-based RV methods in robotic applications, enhancing safety assurance in complex autonomous systems. However, similar to other RV frameworks, RTAMT4ROS solely supports topics monitoring and relies exclusively on ROS. Alternative runtime monitoring systems such as Lola~\cite{DBLP:conf/time/DAngeloSSRFSMM05}, Java-MOP~\cite{DBLP:conf/tacas/ChenR05}, detectEr~\cite{DBLP:journals/scp/AcetoAAEFI24}, Hydra~\cite{DBLP:conf/atva/RaszykBT20}, DejaVu~\cite{DBLP:conf/cpsweek/HavelundPU18}, LamaConv~\cite{lamaconv}, and TraceContract~\cite{DBLP:conf/fm/BarringerH11} could potentially be applied to robotics applications. However, these systems are not explicitly designed for ROS, and integrating them into ROS would require additional development effort and potentially incur runtime costs.

\section{Conclusions and Future Work}


This paper introduces ROSMonitoring 2.0, an extension of the ROSMonitoring framework designed to enable the Runtime Verification of robotic applications developed in ROS. ROSMonitoring 2.0 expands upon its predecessor by facilitating the verification of services, in addition to topics, and by accommodating ordered topics, rather than solely unordered ones. Notably, the new features of ROSMonitoring 2.0 do not necessitate changes to the compositional and formalism-agnostic aspects of ROSMonitoring; only the synthesis of ROS monitors is adjusted. This approach not only leverages all existing features in ROSMonitoring but also ensures full backward compatibility with existing ROS applications based on ROSMonitoring. Furthermore, the proposed ordering algorithm and service interception process hold applicability beyond the scope of ROSMonitoring 2.0, potentially benefiting other systems as well. \marginLD{I would stress that some of this work is not dependent upon ROSMonitoring - e.g., the ordering algorithm could be used by other systems and the process of intercepting services could also be adapted to other systems - hence the work presented here is of general use to other researchers and not just a system description of ROSMonitoring 2.0}\marginMGS{Added a sentence to address this.}

It is also worth noting that the introduction of ordering messages according to the order they are published does not mutually exclude the standard ROSMonitoring topic checking, based on the order the messages are received. In this sense, in ROSMonitoring 2.0 it is also possible to combine both ordering features to monitor both the publish and receive order of messages. This becomes relevant in scenarios where it is necessary to identify which exact node is the faulty one, rather than being only interested in checking the presence of a property violation (which could be relevant in other scenarios instead).

As a future direction, we aim to formally verify that our case study is deadlock-free and establish design principles for ensuring deadlock-freeness. Additionally, threading will be explored as an alternative solution to address potential deadlock issues. We also plan to extend our research to additional case studies in the robotics domain, focusing on complex systems involving multiple services with strong interdependencies across services, topics, and interfaces.

Moreover, we intend to expand the framework to support ROS actions. ROS actions allow robots to execute complex, asynchronous tasks by setting goals, providing feedback, and retrieving results, thus facilitating modular and scalable behaviours for navigation, manipulation, and planning. Although actions are asynchronous and non-blocking, which reduces the monitoring burden compared to services, they introduce challenges in tracking progress against runtime goals.

In parallel, we plan to enhance our message ordering algorithm by introducing timeouts, preventing messages from waiting indefinitely, particularly in unreliable communication scenarios. However, careful consideration is required, as timeouts may disrupt the message order when delays occur, rather than message loss. This could be especially important for scenarios with strict timing requirements, where a balance must be struck between message order and timely delivery.

\marginAF{Paragraph added to discuss a bit more on scalability as a future direction.}
Furthermore, a comprehensive performance evaluation of ROSMonitoring 2.0 will be a critical focus. We aim to assess key metrics such as execution time, resource usage, and system overhead, bench-marking our approach against existing alternatives. Such an evaluation will provide deeper insights into the framework's efficiency and scalability and guide further optimisations.

\marginAF{Porting of reordering left as future direction.}
Lastly, our goal is to port all the features presented in this paper to ROS2, which currently only supports service monitoring and lacks message reordering functionality. This migration will proceed once additional evaluations and testing have been completed on the ROS1 version of ROSMonitoring 2.0.

\bibliographystyle{eptcs}
\bibliography{main}

\end{document}

%% file: comments.tex
\newif\ifcomments
\newif\ifshowremoved

\usepackage{xcolor}
\usepackage[absolute]{textpos}
\usepackage{everypage}

\definecolor{lightblue}{RGB}{210,210,225}
\definecolor{lightred}{RGB}{225,210,210}
\definecolor{lightgreen}{RGB}{210,225,210}
\definecolor{lightyellow}{RGB}{225,222,200}
\definecolor{lightpurple}{RGB}{225,210,225}
\definecolor{warningyellow}{RGB}{247, 245, 187}
\definecolor{darkergreen}{RGB}{0,64,0}
\definecolor{darkred}{RGB}{128,0,0}
\definecolor{darkblue}{RGB}{0,0,139}
\definecolor{darkgreen}{RGB}{0,128,0}
\definecolor{darkpurple}{RGB}{128,0,128}
\definecolor{warningorange}{RGB}{124, 81, 0}
\definecolor{eyecancerpink}{rgb}{1.0, 0.0, 1.0}
\definecolor{radiationyellow}{rgb}{0.8, 1.0, 0.0}

\newcommand{\colorpar}[3]{\colorbox{#1}{\parbox{#2}{#3}}}

\newcommand{\marginremark}[4]{\marginpar{\colorpar{#2}{\linewidth}{\color{#1}\tiny{[#3]~ #4}}}}

\makeatletter
\def\THICKhrulefill{\leavevmode \leaders \hrule height 5pt\hfill \kern \z@}
\makeatother

\newcommand{\marginAF}[1]
{\ifcomments\marginremark{darkblue}{lightblue}{AF}{#1}\fi}
\newcommand{\marginLD}[1]{}
\newcommand{\marginMGS}[1]{}
\newcommand{\marginFMASMGS}[1]{\ifcomments\marginremark{darkgreen}{lightgreen}{MGS}{#1}\fi}

\newcommand{\inlineLD}[1]{}


%% file: rosmonoverview.tex
\tikzset{every picture/.style={line width=0.75pt}} 

\begin{figure}[ht]
\centering

\scalebox{0.9}{

\begin{tikzpicture}[x=0.75pt,y=0.75pt,yscale=-1,xscale=1]

\draw   (90,84) -- (172,84) -- (172,124) -- (90,124) -- cycle ;
\draw  [fill={rgb, 255:red, 224; green, 214; blue, 123 }  ,fill opacity=1 ] (46.1,68) -- (23,68) -- (23,22) -- (56,22) -- (56,58.1) -- cycle -- (46.1,68) ; \draw   (56,58.1) -- (48.08,60.08) -- (46.1,68) ;
\draw    (55,50) .. controls (82.3,48.05) and (69.67,66.06) .. (88.49,82.72) ;
\draw [shift={(90,84)}, rotate = 218.99] [fill={rgb, 255:red, 0; green, 0; blue, 0 }  ][line width=0.75]  [draw opacity=0] (8.93,-4.29) -- (0,0) -- (8.93,4.29) -- cycle    ;


\draw  [fill={rgb, 255:red, 108; green, 154; blue, 209 }  ,fill opacity=1 ] (317.1,127) -- (294,127) -- (294,81) -- (327,81) -- (327,117.1) -- cycle -- (317.1,127) ; \draw   (327,117.1) -- (319.08,119.08) -- (317.1,127) ;
\draw  [fill={rgb, 255:red, 108; green, 154; blue, 209 }  ,fill opacity=1 ] (337.1,147) -- (314,147) -- (314,101) -- (347,101) -- (347,137.1) -- cycle -- (337.1,147) ; \draw   (347,137.1) -- (339.08,139.08) -- (337.1,147) ;
\draw  [fill={rgb, 255:red, 108; green, 154; blue, 209 }  ,fill opacity=1 ] (354.1,175) -- (331,175) -- (331,129) -- (364,129) -- (364,165.1) -- cycle -- (354.1,175) ; \draw   (364,165.1) -- (356.08,167.08) -- (354.1,175) ;
\draw    (172,104) .. controls (199.36,118.93) and (224.25,141.77) .. (302.32,139.04) ;
\draw [shift={(303.5,139)}, rotate = 537.8299999999999] [fill={rgb, 255:red, 0; green, 0; blue, 0 }  ][line width=0.75]  [draw opacity=0] (8.93,-4.29) -- (0,0) -- (8.93,4.29) -- cycle    ;

\draw  [fill={rgb, 255:red, 108; green, 154; blue, 209 }  ,fill opacity=1 ] (304.1,236) -- (281,236) -- (281,190) -- (314,190) -- (314,226.1) -- cycle -- (304.1,236) ; \draw   (314,226.1) -- (306.08,228.08) -- (304.1,236) ;
\draw    (172,104) .. controls (197.37,116.94) and (178.69,193.23) .. (279.47,190.05) ;
\draw [shift={(281,190)}, rotate = 537.77] [fill={rgb, 255:red, 0; green, 0; blue, 0 }  ][line width=0.75]  [draw opacity=0] (8.93,-4.29) -- (0,0) -- (8.93,4.29) -- cycle    ;

\draw  [dash pattern={on 4.5pt off 4.5pt}] (230,34.5) -- (417.5,34.5) -- (417.5,263) -- (230,263) -- cycle ;

\draw    (315.5,201) .. controls (364.01,214.86) and (381.16,123.85) .. (444.56,121.06) ;
\draw [shift={(446.5,121)}, rotate = 539.12] [fill={rgb, 255:red, 0; green, 0; blue, 0 }  ][line width=0.75]  [draw opacity=0] (8.93,-4.29) -- (0,0) -- (8.93,4.29) -- cycle    ;

\draw  [fill={rgb, 255:red, 194; green, 108; blue, 214 }  ,fill opacity=1 ] (469.1,133) -- (446,133) -- (446,87) -- (479,87) -- (479,123.1) -- cycle -- (469.1,133) ; \draw   (479,123.1) -- (471.08,125.08) -- (469.1,133) ;
\draw   (467,192.5) .. controls (467,179.52) and (490.95,169) .. (520.5,169) .. controls (550.05,169) and (574,179.52) .. (574,192.5) .. controls (574,205.48) and (550.05,216) .. (520.5,216) .. controls (490.95,216) and (467,205.48) .. (467,192.5) -- cycle ;
\draw    (316.53,215.81) .. controls (382.37,241.21) and (409.05,167.76) .. (465.29,191.74) ;
\draw [shift={(467,192.5)}, rotate = 204.56] [fill={rgb, 255:red, 0; green, 0; blue, 0 }  ][line width=0.75]  [draw opacity=0] (8.93,-4.29) -- (0,0) -- (8.93,4.29) -- cycle    ;
\draw [shift={(314.5,215)}, rotate = 22.38] [fill={rgb, 255:red, 0; green, 0; blue, 0 }  ][line width=0.75]  [draw opacity=0] (8.93,-4.29) -- (0,0) -- (8.93,4.29) -- cycle    ;
\draw  [fill={rgb, 255:red, 119; green, 221; blue, 197 }  ,fill opacity=1 ] (458.1,274) -- (435,274) -- (435,228) -- (468,228) -- (468,264.1) -- cycle -- (458.1,274) ; \draw   (468,264.1) -- (460.08,266.08) -- (458.1,274) ;
\draw    (520.41,218.31) .. controls (519.83,242.69) and (531.75,257.34) .. (468.93,246.34) ;
\draw [shift={(467,246)}, rotate = 370.15] [fill={rgb, 255:red, 0; green, 0; blue, 0 }  ][line width=0.75]  [draw opacity=0] (8.93,-4.29) -- (0,0) -- (8.93,4.29) -- cycle    ;
\draw [shift={(520.5,216)}, rotate = 93.3] [fill={rgb, 255:red, 0; green, 0; blue, 0 }  ][line width=0.75]  [draw opacity=0] (8.93,-4.29) -- (0,0) -- (8.93,4.29) -- cycle    ;

\draw    (479,109) .. controls (527.02,122.72) and (471.32,147) .. (487.89,170.56) ;
\draw [shift={(489,172)}, rotate = 230.19] [fill={rgb, 255:red, 0; green, 0; blue, 0 }  ][line width=0.75]  [draw opacity=0] (8.93,-4.29) -- (0,0) -- (8.93,4.29) -- cycle    ;

\draw    (298.85,186.06) .. controls (285.58,145.79) and (303.42,154.29) .. (319.73,153.16) ;
\draw [shift={(321.5,153)}, rotate = 533.29] [fill={rgb, 255:red, 0; green, 0; blue, 0 }  ][line width=0.75]  [draw opacity=0] (8.93,-4.29) -- (0,0) -- (8.93,4.29) -- cycle    ;
\draw [shift={(299.5,188)}, rotate = 251.18] [fill={rgb, 255:red, 0; green, 0; blue, 0 }  ][line width=0.75]  [draw opacity=0] (8.93,-4.29) -- (0,0) -- (8.93,4.29) -- cycle    ;

\draw (131,104) node  [scale=1.3] [align=center] {instrument};
\draw (40,9) node [scale=1.5] [align=left] {config.yaml};
\draw (318,62) node [scale=1.5] [align=center] {nodes};
\draw (296,249) node [scale=1.5] [align=left] {monitor.py};
\draw (398,22) node  [scale=1.5] [align=left] {ROS};
\draw (462,74) node [scale=1.5] [align=left] {log.txt};
\draw (520.5,192.5) node  [scale=1.3] [align=center] {oracle};
\draw (450,215) node [scale=1.5] [align=left] {spec};
\draw (442,177) node [scale=1.5] [align=left] {online};
\draw (515,151) node [scale=1.5] [align=left] {offline};

\end{tikzpicture}

}
\caption{High-level overview of ROSMonitoring~\cite{DBLP:conf/taros/FerrandoC0AFM20}.}
\label{fig:rosmon-pipeline}
\end{figure}

%% file: example.tex
\begin{tikzpicture}[
     node/.style={rectangle, rounded corners=3mm, draw, minimum size=1cm, inner sep=10pt},
    topic/.style={draw, minimum width=1.5cm, minimum height=.5cm, inner sep=2pt},
    arrow/.style={-Stealth, thick},
    dashedarrow/.style={-Stealth, dashed, thick},
    doublearrow/.style={thick, >=Stealth, <->}
    ]

    \node[node] at (0,0) (battery) {Battery};
    \node[topic] at (0, -2) (percentage) {\textbackslash battery\uline{0.2cm}percentage};
    \node[node] at (5, -2)  (batterymon) {Battery Supervisor};
    \node[topic] at (10,-2) (batterystat) {\textbackslash battery\uline{0.2cm}status};
    \node[node] at (5,0) (ledpanel) {LED Panel};
    \node[topic] at (10,0) (ledpaneltopic) {\textbackslash LED\uline{0.2cm}Panel};


    \draw[arrow] (batterymon) -- (batterystat) node[midway, above] {};
    \draw[arrow] (battery) -- (percentage) node[midway, above] {};
    \draw[arrow] (percentage) -- (batterymon) node[midway, above] {};
    \draw[arrow] (ledpanel) -- (ledpaneltopic) node[midway, above] {};
    \draw[dashedarrow] (batterymon) -- node[midway, left] {\textbackslash SetLED invocation} (ledpanel);

\end{tikzpicture}

%% file: seq-diagram.tex
\tikzset{every picture/.style={line width=0.75pt}} 

\begin{figure}[!ht]
\centering

\scalebox{0.7}{
\begin{tikzpicture}[x=0.75pt,y=0.75pt,yscale=-1,xscale=1]

\draw  [dash pattern={on 4.5pt off 4.5pt}]  (720,49) -- (719.5,622) ;
\draw  [dash pattern={on 4.5pt off 4.5pt}]  (473,50) -- (472.5,623) ;
\draw  [dash pattern={on 4.5pt off 4.5pt}]  (237,49) -- (236.5,622) ;
\draw  [dash pattern={on 4.5pt off 4.5pt}]  (75,49) -- (74.5,622) ;
\draw   (425,9) -- (525.5,9) -- (525.5,49) -- (425,49) -- cycle ;
\draw   (671,9) -- (771.5,9) -- (771.5,49) -- (671,49) -- cycle ;
\draw  [fill={rgb, 255:red, 255; green, 255; blue, 255 }  ,fill opacity=1 ] (707,108) -- (732.5,108) -- (732.5,190) -- (707,190) -- cycle ;
\draw  [dash pattern={on 4.5pt off 4.5pt}]  (704.5,133.97) -- (502.5,131.8) ;
\draw [shift={(707.5,134)}, rotate = 180.62] [fill={rgb, 255:red, 0; green, 0; blue, 0 }  ][line width=0.08]  [draw opacity=0] (8.93,-4.29) -- (0,0) -- (8.93,4.29) -- cycle    ;
\draw    (452.5,87.99) -- (88.5,86.8) ;
\draw [shift={(455.5,88)}, rotate = 180.19] [fill={rgb, 255:red, 0; green, 0; blue, 0 }  ][line width=0.08]  [draw opacity=0] (8.93,-4.29) -- (0,0) -- (8.93,4.29) -- cycle    ;
\draw  [fill={rgb, 255:red, 255; green, 255; blue, 255 }  ,fill opacity=1 ] (461,64.8) -- (486.5,64.8) -- (486.5,151) -- (461,151) -- cycle ;
\draw   (26,9) -- (126.5,9) -- (126.5,49) -- (26,49) -- cycle ;
\draw  [fill={rgb, 255:red, 255; green, 255; blue, 255 }  ,fill opacity=1 ] (62,64.8) -- (87.5,64.8) -- (87.5,597) -- (62,597) -- cycle ;
\draw   (188,9) -- (288.5,9) -- (288.5,49) -- (188,49) -- cycle ;
\draw  [fill={rgb, 255:red, 255; green, 255; blue, 255 }  ,fill opacity=1 ] (477,88.8) -- (502.5,88.8) -- (502.5,133.8) -- (477,133.8) -- cycle ;
\draw    (502.5,95.8) -- (530.5,95.8) -- (530.5,125.8) -- (504.5,124.9) ;
\draw [shift={(501.5,124.8)}, rotate = 1.97] [fill={rgb, 255:red, 0; green, 0; blue, 0 }  ][line width=0.08]  [draw opacity=0] (8.93,-4.29) -- (0,0) -- (8.93,4.29) -- cycle    ;

\draw  [dash pattern={on 4.5pt off 4.5pt}]  (706,184) -- (494.5,183.01) ;
\draw [shift={(491.5,183)}, rotate = 0.27] [fill={rgb, 255:red, 0; green, 0; blue, 0 }  ][line width=0.08]  [draw opacity=0] (8.93,-4.29) -- (0,0) -- (8.93,4.29) -- cycle    ;
\draw  [fill={rgb, 255:red, 255; green, 255; blue, 255 }  ,fill opacity=1 ] (461,177.8) -- (486.5,177.8) -- (486.5,401) -- (461,401) -- cycle ;
\draw   (8.13,195.23) -- (660.72,195.23) -- (660.72,379) -- (8.13,379) -- cycle ;
\draw  [dash pattern={on 4.5pt off 4.5pt}]  (660.72,265.37) -- (10.13,265.83) ;
\draw    (9.73,210.76) -- (43.08,210.64) ;
\draw    (43.08,210.64) -- (52.13,201.94) ;
\draw    (52.13,201.94) -- (52.13,195.47) ;
\draw    (459.5,226) -- (95.5,224.81) ;
\draw [shift={(92.5,224.8)}, rotate = 0.19] [fill={rgb, 255:red, 0; green, 0; blue, 0 }  ][line width=0.08]  [draw opacity=0] (8.93,-4.29) -- (0,0) -- (8.93,4.29) -- cycle    ;
\draw  [fill={rgb, 255:red, 255; green, 255; blue, 255 }  ,fill opacity=1 ] (474,206.8) -- (499.5,206.8) -- (499.5,251.8) -- (474,251.8) -- cycle ;
\draw    (499.5,213.8) -- (527.5,213.8) -- (527.5,243.8) -- (501.5,242.9) ;
\draw [shift={(498.5,242.8)}, rotate = 1.97] [fill={rgb, 255:red, 0; green, 0; blue, 0 }  ][line width=0.08]  [draw opacity=0] (8.93,-4.29) -- (0,0) -- (8.93,4.29) -- cycle    ;

\draw    (461.5,299) -- (255.5,299) ;
\draw [shift={(252.5,299)}, rotate = 360] [fill={rgb, 255:red, 0; green, 0; blue, 0 }  ][line width=0.08]  [draw opacity=0] (8.93,-4.29) -- (0,0) -- (8.93,4.29) -- cycle    ;
\draw  [fill={rgb, 255:red, 255; green, 255; blue, 255 }  ,fill opacity=1 ] (224,275) -- (249.5,275) -- (249.5,368.8) -- (224,368.8) -- cycle ;
\draw    (456.5,340) -- (251.5,339.8) ;
\draw [shift={(459.5,340)}, rotate = 180.06] [fill={rgb, 255:red, 0; green, 0; blue, 0 }  ][line width=0.08]  [draw opacity=0] (8.93,-4.29) -- (0,0) -- (8.93,4.29) -- cycle    ;
\draw  [dash pattern={on 4.5pt off 4.5pt}]  (700.5,364) -- (486.5,363.8) ;
\draw [shift={(703.5,364)}, rotate = 180.05] [fill={rgb, 255:red, 0; green, 0; blue, 0 }  ][line width=0.08]  [draw opacity=0] (8.93,-4.29) -- (0,0) -- (8.93,4.29) -- cycle    ;
\draw  [fill={rgb, 255:red, 255; green, 255; blue, 255 }  ,fill opacity=1 ] (709,357) -- (734.5,357) -- (734.5,439) -- (709,439) -- cycle ;
\draw  [dash pattern={on 4.5pt off 4.5pt}]  (709,433) -- (497.5,432.01) ;
\draw [shift={(494.5,432)}, rotate = 0.27] [fill={rgb, 255:red, 0; green, 0; blue, 0 }  ][line width=0.08]  [draw opacity=0] (8.93,-4.29) -- (0,0) -- (8.93,4.29) -- cycle    ;
\draw  [fill={rgb, 255:red, 255; green, 255; blue, 255 }  ,fill opacity=1 ] (461,425.8) -- (486.5,425.8) -- (486.5,597) -- (461,597) -- cycle ;
\draw   (8.13,443.23) -- (660.72,443.23) -- (660.72,575) -- (8.13,575) -- cycle ;
\draw  [dash pattern={on 4.5pt off 4.5pt}]  (660.72,513.37) -- (11.13,512.83) ;
\draw    (8.73,458.76) -- (42.08,458.64) ;
\draw    (42.08,458.64) -- (51.13,450.94) ;
\draw    (51.13,450.94) -- (51.13,444.47) ;
\draw  [fill={rgb, 255:red, 255; green, 255; blue, 255 }  ,fill opacity=1 ] (474,454.8) -- (499.5,454.8) -- (499.5,499.8) -- (474,499.8) -- cycle ;
\draw    (499.5,461.8) -- (527.5,461.8) -- (527.5,491.8) -- (501.5,490.9) ;
\draw [shift={(498.5,490.8)}, rotate = 1.97] [fill={rgb, 255:red, 0; green, 0; blue, 0 }  ][line width=0.08]  [draw opacity=0] (8.93,-4.29) -- (0,0) -- (8.93,4.29) -- cycle    ;

\draw    (459.5,479) -- (95.5,477.81) ;
\draw [shift={(92.5,477.8)}, rotate = 0.19] [fill={rgb, 255:red, 0; green, 0; blue, 0 }  ][line width=0.08]  [draw opacity=0] (8.93,-4.29) -- (0,0) -- (8.93,4.29) -- cycle    ;
\draw    (460.5,546) -- (96.5,544.81) ;
\draw [shift={(93.5,544.8)}, rotate = 0.19] [fill={rgb, 255:red, 0; green, 0; blue, 0 }  ][line width=0.08]  [draw opacity=0] (8.93,-4.29) -- (0,0) -- (8.93,4.29) -- cycle    ;
\draw  [fill={rgb, 255:red, 255; green, 255; blue, 255 }  ,fill opacity=1 ] (724,130.8) -- (749.5,130.8) -- (749.5,175.8) -- (724,175.8) -- cycle ;
\draw    (749.5,137.8) -- (777.5,137.8) -- (777.5,167.8) -- (751.5,166.9) ;
\draw [shift={(748.5,166.8)}, rotate = 1.97] [fill={rgb, 255:red, 0; green, 0; blue, 0 }  ][line width=0.08]  [draw opacity=0] (8.93,-4.29) -- (0,0) -- (8.93,4.29) -- cycle    ;

\draw  [fill={rgb, 255:red, 255; green, 255; blue, 255 }  ,fill opacity=1 ] (728,374.8) -- (753.5,374.8) -- (753.5,419.8) -- (728,419.8) -- cycle ;
\draw    (753.5,381.8) -- (781.5,381.8) -- (781.5,411.8) -- (755.5,410.9) ;
\draw [shift={(752.5,410.8)}, rotate = 1.97] [fill={rgb, 255:red, 0; green, 0; blue, 0 }  ][line width=0.08]  [draw opacity=0] (8.93,-4.29) -- (0,0) -- (8.93,4.29) -- cycle    ;

\draw (475.25,29) node   [align=left] {ROSMonitor};
\draw (721.25,29) node   [align=left] {Oracle};
\draw (76.25,29) node   [align=left] {Client Node};
\draw (238.25,29) node   [align=left] {Service Node};
\draw (107,68) node [anchor=north west][inner sep=0.75pt]  [font=\small] [align=left] {callService(req, res)};
\draw (505,75) node [anchor=north west][inner sep=0.75pt]  [font=\small] [align=left] {callbackService(req, res)};
\draw (561,113) node [anchor=north west][inner sep=0.75pt]  [font=\small] [align=left] {sendRequest(req)};
\draw (550,162) node [anchor=north west][inner sep=0.75pt]  [font=\small] [align=left] {verdictOn(req)};
\draw (26.99,202) node   [align=left] {alt};
\draw (597,198) node [anchor=north west][inner sep=0.75pt]  [font=\small] [align=left] {if violation};
\draw (599,268) node [anchor=north west][inner sep=0.75pt]  [font=\small] [align=left] {otherwise};
\draw (109,203) node [anchor=north west][inner sep=0.75pt]  [font=\small] [align=left] {error on callService};
\draw (529,222) node [anchor=north west][inner sep=0.75pt]  [font=\small] [align=left] {publishError()};
\draw (284,279) node [anchor=north west][inner sep=0.75pt]  [font=\small] [align=left] {callService(req, res)};
\draw (301,319) node [anchor=north west][inner sep=0.75pt]  [font=\small] [align=left] {response(res)};
\draw (502,342) node [anchor=north west][inner sep=0.75pt]  [font=\small] [align=left] {sendResponse(res)};
\draw (553,411) node [anchor=north west][inner sep=0.75pt]  [font=\small] [align=left] {verdictOn(res)};
\draw (26.99,450) node   [align=left] {alt};
\draw (597,446) node [anchor=north west][inner sep=0.75pt]  [font=\small] [align=left] {if violation};
\draw (599,516) node [anchor=north west][inner sep=0.75pt]  [font=\small] [align=left] {otherwise};
\draw (529,470) node [anchor=north west][inner sep=0.75pt]  [font=\small] [align=left] {publishError()};
\draw (109,456) node [anchor=north west][inner sep=0.75pt]  [font=\small] [align=left] {error on callService};
\draw (110,523) node [anchor=north west][inner sep=0.75pt]  [font=\small] [align=left] {response(res)};
\draw (779.5,140.8) node [anchor=north west][inner sep=0.75pt]  [font=\small] [align=left] {check request};
\draw (783.5,384.8) node [anchor=north west][inner sep=0.75pt]  [font=\small] [align=left] {check response};

\end{tikzpicture}

}
\caption{Service verification in ROSMonitoring 2.0 when filtering is enabled.}
\label{fig:rosmon-service-verification}
\end{figure}

%% file: casestudy.tex
\begin{tikzpicture}[
    node/.style={rectangle, rounded corners=3mm, draw, minimum size=1cm, inner sep=10pt},
    tallnode/.style={rectangle, rounded corners=3mm, draw, minimum size=1cm, inner sep=10pt, minimum height=1.5cm},
    topic/.style={draw, minimum width=1.5cm, minimum height=.5cm, inner sep=2pt},
    arrow/.style={-Stealth, thick},
    dashedarrow/.style={-Stealth, dashed, thick},
    doublearrow/.style={thick, >=Stealth, <->}
    ]

    \node[node] at (0,0) (battery) {Battery};
    \node[topic] at (0, -2) (percentage) {\textbackslash battery\uline{0.2cm}percentage};
    \node[node] at (5, -2)  (batterymon) {Battery Supervisor};
    \node[topic] at (10,-1.25) (inputacc) {\textbackslash input\uline{0.2cm}accepted};
    \node[topic] at (10,-2) (batterystat) {\textbackslash battery\uline{0.2cm}status};
    \node[topic] at (10,-3) (statuschange) {\textbackslash status\uline{0.2cm}change};
    \node[tallnode] at (5,0) (rosmon) {ROS Monitor};
    \node[node] at (5,2) (ledpanel) {LED Panel};
    \node[topic] at (10,2.75) (ledpaneltopic) {\textbackslash LED\uline{0.2cm}Panel};
    \node[topic] at (2, 0) (verdict) {\textbackslash verdict};
    \node[topic] at (10,2) (statusacc) {\textbackslash status\uline{0.2cm}accepted};



    \draw[arrow] (batterymon)+(.2,.5) |- (inputacc);
    \path[arrow, draw, transform canvas = {yshift=-.4cm}] (inputacc.east)+(0,0.4) -| ($(inputacc)+(1.75,1.25)$) |- (rosmon.east);
    \draw[arrow] (batterymon) -- (batterystat);
    \path[arrow, draw, transform canvas = {yshift=-.2cm}] (batterystat.east)+(0,.2) -| ($(batterystat)+(2,2)$) |- (rosmon.east);
    \draw[doublearrow] (batterymon)+(.2,-.5) |- (statuschange) node[midway, above] {};
    \draw[arrow] (battery) -- (percentage) node[midway, above] {};
    \draw[arrow] (percentage) -- (batterymon) node[midway, above] {};
    \draw[arrow] (rosmon) -- (verdict) node[midway, above] {};
    \draw[arrow] (ledpanel.north) |- ($(ledpaneltopic)+(-2,0)$) -- (ledpaneltopic.west);
     \draw[arrow] (ledpanel)-- (statusacc);
    \path[arrow, draw, transform canvas = {yshift=0.2cm}] (ledpaneltopic.east)+(0,-0.2) -| ($(ledpaneltopic)+(2.25,-2)$) |- (rosmon.east);
    \path[arrow, draw, transform canvas = {yshift=0.4cm}] (statusacc.east)+(0,-0.4) -| ($(statusacc)+(2,-1.75)$) |- (rosmon.east);
    \draw[dashedarrow] (batterymon) -- node[midway, left] {\textbackslash SetLED\uline{0.2cm}mon invocation} (rosmon);
    \draw[dashedarrow] (rosmon) -- node[midway, left] {\textbackslash SetLED invocation} (ledpanel);

\end{tikzpicture}